\documentclass[
]{ceurart}

\usepackage[utf8]{inputenc}
\usepackage[T1]{fontenc}
\usepackage{ceurstuffthm}
\usepackage{cwdefs}
\usepackage{pgfkeys}
\usepackage{plabels}
\usepackage{array}
\usepackage{longtable}

\usepackage{listings}
\newcommand{\LL}{\mathsf{L}}
\newcommand{\RR}{\mathsf{R}}

\newcommand{\seq}{\Rightarrow}
\newcommand{\seqh}[1]{\;\stackrel{#1}{\Longrightarrow}\;}

\newcommand{\tabrulename}[1]{\raisebox{-2.0ex}[0pt][0pt]{\small #1}}

\newcommand{\rulename}[1]{\textsc{#1}}

\renewcommand{\X}{\mathrm{\Gamma}}
\renewcommand{\Y}{\mathrm{\Delta}}
\newcommand{\AX}{\bigwedge\mathrm{\Gamma}}
\newcommand{\VY}{\bigvee\mathrm{\Delta}}

\newcommand{\htform}{HT-formula\xspace}
\newcommand{\htforms}{HT-formulas\xspace}
\newcommand{\nhnnf}{NHNNF-formula\xspace}
\newcommand{\nhnnfs}{NHNNF-formulas\xspace}

\newcommand{\nh}{\f{nh}}
\newcommand{\vocpn}{\f{voc}^\pm}

\newcommand{\NF}{\f{NF}}
\newcommand{\F}{\f{F}}
\newcommand{\T}{\f{T}}

\newcommand{\xequiv}{\;\equiv\;}

\begin{document}

\setlength\abovedisplayshortskip{0.0pt plus 1.0pt minus 1.0pt}
\setlength\abovedisplayskip{9.0pt plus 3.0pt minus 5.0pt}
\setlength\belowdisplayskip{9.0pt plus 3.0pt minus 5.0pt}

\copyrightyear{2026}
\copyrightclause{Copyright for this paper by its author.
  Use permitted under Creative Commons License Attribution 4.0
  International (CC BY 4.0).}

\copyrightclause{Copyright for this paper by its author.}

\conference{CI-BD-SOQE 2026: Workshop on Craig Interpolation, Beth
  Definability, and Second-Order Quantifier Elimination, at FLoC 2026: The
  9th Federated Logic Conference, July 24--25, 2026, Lisbon, Portugal}

\ccLogoWidth=0.01pt

\title{Craig-Lyndon Interpolation for the Logic of Here and There with a Variation of
  Mints' Sequent System\\ (Extended Abstract)}

\author{Christoph Wernhard}

\address{University of Potsdam, Germany}

\begin{abstract}
  We present a variation of Maehara's method to construct Craig-Lyndon
  interpolants for the three-valued propositional logic of here and there
  (HT), also known as Gödel's $G_3$, a superintuitionistic logic of importance
  in logic programming. Our method adapts a recent interpolation technique
  that operates on classically encoded logic programs to a variation of Mints'
  sequent system for HT. The approach is characterized by two stages: First, a
  preliminary interpolant is constructed, a formula that is an interpolant in
  some sense but not yet the desired HT formula. In the second stage, an
  actual HT interpolant is obtained from this preliminary interpolant. With
  the classical encoding, the preliminary interpolant is a classical
  Craig-Lyndon interpolant for classical encodings of the two input HT
  formulas. In the presented adaptation, the sequent system operates directly
  on HT formulas, and the preliminary interpolant is in a nonclassical logic
  that generalizes HT by an additional logic operator.
\end{abstract}


\maketitle

\section{Introduction}

The propositional \name{logic of here and there (HT)}, also known as Gödel's
$G_3$, is a superintuitionistic, or intermediate, logic, which plays an
important role in logic programming. Central to this role is the observation,
due to Lifschitz, Pearce and Valverde \cite{LifschitzEtAl2001}, that for
answer set programs under the stable model semantics a particularly useful
notion of program equivalence, \name{strong equivalence}, can be characterized
as equivalence in HT.
We consider here \emph{Craig-Lyndon interpolation} for HT, aiming at
applications in knowledge representation and program synthesis such as
described, e.g., in
\cite{gpv:interpolable:2011,pearce:valverde:synonymous:2012,hw:2024}. As noted
in \cite{gpv:interpolable:2011}, results by Maksimova
\cite{maksimova:superintutionistic:1977,maksimova:superintutionistic:1979,gabbay:maximova:book}
prove that HT has the Craig interpolation property, i.e., for two HT formulas
such that one can be inferred from the other, a Craig interpolant exists.
Baaz and Veith \cite[Section 4.1.1]{baaz:veith:fuzzy:1999} have shown that
uniform interpolants for HT can be constructed as the disjunction over the
formulas obtained from the given formula by substituting the atoms to
eliminate in all possible ways with atoms to be retained and truth value
constants for true and false.

Recently \cite{hw:2024}, Craig-Lyndon interpolation for HT was proven with a
practically applicable construction of HT interpolants from clausal tableau
proofs, which can be obtained from theorem provers for classical first-order
logic. The underlying proof system, clausal tableaux, is for classical
first-order logic and it is applied to HT formulas under a specific encoding
into classical logic. This encoding is well established as a technique for
mechanically proving strong equivalence of two given logic programs with
theorem provers for classical logic. The core idea of the encoding is to
represent an HT predicate by two classical predicates, one for the ``here''
world and the other for the ``there'' world.

Since \cite{hw:2024} mainly addresses logic programming, HT interpolation is
expressed in that paper as ``LP-interpolation'' on the level of first-order
formulas that encode HT formulas \cite[Theorem~11]{hw:2024}. The transfer to
interpolation directly on ``logic programs'', i.e., HT formulas in a certain
normal form, is straightforward and made explicit in the Prolog implementation
that accompanies the paper (predicate \texttt{p\_interpolant/4}). Since any HT
formula can be normalized into the special case of a ``logic program'', as
outlined by Mints \cite{mints}, the method of \cite{hw:2024} is applicable
w.l.o.g. to HT in full. \pagebreak Moreover, the method from \cite{hw:2024} actually
applies not just to propositional HT but also to a first-order generalization
of HT. As outlined in \cite{hw:2024} and also implemented, it can operate
either directly on clausal first-order tableaux or on resolution proofs, which
are converted to clausal tableaux.

With respect to practical application, the method of \cite{hw:2024} provides
exactly what is needed: powerful and highly optimized stock first-order
provers can be utilized to compute HT interpolants. However, the method does
not fit immediately into research threads that are centered around
\emph{sequent systems}. Hence, our objective here is the transfer of the
method to the sequent system perspective, without involvement of a classical
encoding.

With the method, an HT interpolant is constructed in two stages: first, a
preliminary interpolant is constructed, a formula that is an interpolant in
some sense but not yet the desired HT formula. In the second stage, an actual
HT interpolant is obtained from this preliminary interpolant.
In the ``classical encoding variation'' of the method, as described in
\cite{hw:2024}, the preliminary interpolant is a classical Craig-Lyndon
interpolant for classical encodings of the two input HT formulas. A
postprocessing operation then yields the classical encoding of an HT formula
that is an interpolant of the original HT inputs. Decoding the latter formula
into an actual HT formula is then a trivial matter.
In the ``sequent system variation'', described here, the input HT formulas are
directly processed by a sequent system for HT, a modification of the
propositional fragment of Mints' system for HT \cite{mints}. The sequent proof
allows calculation of a preliminary interpolant in a logic that generalizes
HT. The postprocessing operation then converts this to a proper HT
interpolant.

The rest of this short paper is structured as follows: We specify the
considered logics (Sect.~\ref{sec-considered-logics}), and present axioms and
rules that extend Mints' system (Sect.~\ref{sec-axioms-and-rules}). After
describing Maehara's interpolation method (Sect.~\ref{sec-maehara-intro}), we
use it to define an interpolating sequent system that constructs for given HT
formulas an interpolant in a logic that extends HT
(Sect.~\ref{sec-interpolating-nhnnf}). We then show, how this preliminary
interpolant can be converted to the desired HT interpolant
(Sect.~\ref{sec-ht-interpolation}). We conclude by noting issues for further
investigation (Sect.~\ref{sec-conclusion}).
Examples of HT interpolation applied to synthesize strongly equivalent logic
programs are provided in \cite[Examples 13 and 16]{hw:2024}.

\section{Logics}
\label{sec-considered-logics}

We consider the three-valued propositional logic of here and there (HT) and
variations of this logic. Truth values are $\F$ (\name{false}), $\T$
(\name{true}), and $\NF$ (\name{not false}). Logic operators are the
well-known $\false,\true,\lnot,\lor,\land,\imp$ along with the unary operator
$\nh$ (,,\name{not here}''). Their semantics is specified with the truth
tables in Table~\ref{table-truth-tables}.
A \defname{literal} is either an atom or the
negation $\lnot p$ of an atom $p$.
The \defname{polarity} of a subformula occurrence is \defname{positive}
(\defname{negative}) if it is in the scope of an even (odd) number of
occurrences of $\lnot$, $\nh$, and, with respect to its first argument,
$\imp$. The \defname{polarity-aware vocabulary} $\vocpn(A)$ of a formula $A$
is defined as the set of all pairs $\la +, p\ra$ where $p$ is an atom with a
positive occurrence in $A$ and all pairs $\la -, p\ra$ where $p$ is an atom
with a negative occurrence in $A$. For example, $\vocpn((\fp \land \lnot \fq)
\imp (\nh(\fr) \lor \fs \lor \lnot s)) = \{\la -, p\ra, \la +, q\ra, \la -,
r\ra, \la +, s\ra, \la -, s\ra\}$.

We quietly assume that semantic assertions with free formula parameters hold
for ``all'' formulas, which includes any formulas that could be built with an
extended truth-functionally complete set of operators.
Our semantics allows the synonymous use of: (a)~$A$ entails $B$; (b)~$A
\entails B$; (c)~$A \imp B$ is valid; (d)~Under all assignments of the atoms
in $A$ or $B$ to truth values from $\{\F, \NF, \T\}$ the truth value of $A
\imp B$ is $\T$.\footnote{\label{foot-entailment}Entailment is understood
\emph{comparative} \cite{bendova:2005}: $A \entails B$ iff under all truth
value assignments the value of $A$ is less than or equal to that of $B$, where
$\F \leq \NF \leq \T$. For HT, which has a classical deduction theorem, this
is equivalent to \emph{preservation of the designated value}: under all truth
value assignments, if the value of $A$ is $\T$, then that of $B$ is $\T$.}
We use $A \equi B$ to abbreviate $(A \imp B) \land (B \imp A)$ and
synonymously use: (a)~$A$ and $B$ are equivalent; (b)~$A \equiv B$; (c)~$A
\equi B$ is valid; (d)~Under all assignments of the atoms in $A$ or $B$ to
truth values from $\{\F, \NF, \T\}$ the truth value of $A \equi B$ is $\T$;
(e)~Under all assignments of the atoms in $A$ or $B$ to truth values from
$\{\F, \NF, \T\}$ the truth value of $A$ is the same as the truth value of
$B$.

\begin{table}
\caption{Truth tables for the considered logic operators. They extend the
  valuations for HT by the specification of $\nh(A)$.}
\label{table-truth-tables}
\centering
  $\begin{array}[t]{c@{\hspace{1em}}c}
       \false & \true\\\midrule
       \F    &  \T
     \end{array}
     \hspace{2em}
     \begin{array}[t]{c|@{\hspace{1em}}c@{\hspace{1em}}c}
       A & \lnot A & \nh(A)\\\midrule
       \F     & \T   &  \T\\
       \NF    & \F   &  \T\\
       \T     & \F   &  \F\\
     \end{array}
     \hspace{2em}
\begin{array}[t]{cc|@{\hspace{1em}}c@{\hspace{1em}}c@{\hspace{1em}}c}
  A & B  & A \lor B & A \land B & A \imp B\\\midrule
  \F & \F  &   \F      &  \F        & \T\\
  \F & \NF &   \NF     &  \F        & \T\\
  \F & \T  &   \T      &  \F        & \T\\
  \NF & \F &   \NF     &  \F        & \F\\
  \NF & \NF &  \NF     & \NF        & \T\\
  \NF & \T &   \T      & \NF        & \T\\
  \T & \F  &   \T      &  \F        & \F\\
  \T & \NF &   \T      & \NF        & \NF\\
  \T & \T  &   \T      &  \T        & \T
\end{array}$
\end{table}

\pagebreak
The formulas in which we are ultimately interested and for which we will
develop our main result on Craig-Lyndon interpolation are those of HT, defined
with the following grammar.

\medskip
{\setlength{\tabcolsep}{0pt}
\begin{tabular}{p{8em}l}
\defname{\htform}: & $A, B := \g{Atom} \mid \false \mid
\true \mid \lnot A  \mid (A \lor B) \mid (A \land B) \mid (A \imp B).$
\end{tabular}
}
\medskip

\noindent
We include the truth value constants $\true, \false$, which does not extend
expressiveness, as $\false \equiv (A \land \lnot A) \equiv \lnot \true$ and
$\true \equiv \lnot \false$. At least one of $\true$ or $\false$ is required
to take account of interpolation for formulas with disjoint vocabularies.
Our construction of preliminary interpolants will involve a further formula
class, without $\imp$ but with $\nh$ and in a negation normal form:

\medskip
{\setlength{\tabcolsep}{0pt}
\begin{tabular}{p{8em}l}
\defname{\nhnnf}: & $A, B := \g{Atom} \mid \lnot \g{Atom} \mid
\lnot \lnot \g{Atom} \mid \nh(\g{Atom}) \mid \false \mid \true \mid (A \lor B)
\mid (A \land B).$
\end{tabular}
}
\medskip

Intuitively, $\nh(A)$ is true (i.e., $\T$) if and only if $A$ is false in the
\name{here} world, independently of its value in the \name{there} world, and
false (i.e., $\F$) otherwise. We consider the $\nh$ operator because a
variation of Maehara's method allows from proofs for HT-formulas the
extraction of Craig-Lyndon interpolants that are \nhnnfs. In a postprocessing
step these can be converted to \htforms. The $\nh$ operator has the following
properties.
\begin{align}
\lnot A \imp \nh(A) & \xequiv \true.\\
\nh(A) \lor A & \xequiv \true.\\
  \label{eq-nh-in-not}
  \nh(\lnot A) & \xequiv \lnot \lnot A.\\
  \label{eq-nh-in-and}
  \nh(A \land B) & \xequiv \nh(A) \lor \nh(B).\\
  \label{eq-nh-in-or}
\nh(A \lor B) & \xequiv \nh(A) \land \nh(B).\\
\label{eq-imp-nh}  
A \imp B & \xequiv (\nh(A) \lor B) \land (\lnot A \lor \lnot \lnot B).
\end{align}

\section{Axioms and Rules}
\label{sec-axioms-and-rules}

Our starting point is the propositional fragment of Mints' sequent system for
HT \cite{mints}. We assume straightforward restrictions of its axioms and also
consider some additional axioms and rules.

\paragraph{Restrictions of Mints' Axioms.}
\label{sec-mints-axioms}

We assume the two axioms of Mints' system for HT with the annotated
restrictions. 
\[
\begin{array}{l@{\hspace{2em}}l@{\hspace{2em}}l}
{\small \rulename{Ax-1}} & A, \X \seq \Y, A & A \text{ a literal}\\
{\small \rulename{Ax-2}} & A, \lnot A, \X \seq \Y & A \text{ an atom}
\end{array}
\]

\paragraph{Axioms and Rules for Truth Value Constants.}

Axioms and rules for handling $\top$ and $\false$ are straightforward and
omitted in the presentation.

\paragraph{Axioms and Rules for $\nh$.}
\label{sec-axrules-nh}
We have the following axioms and rule.
\[
\begin{array}{l@{\hspace{2em}}l@{\hspace{2em}}l}
  {\small \rulename{Ax-nh-1}} & \X \seq \Y, A, \nh(A)  & A \text{ an atom}\\
  {\small \rulename{Ax-nh-2}} & \lnot A, \X \seq \Y, \nh(A) & A \text{ an atom}\\
\end{array}
\]
\[
\begin{array}{r@{\hspace{2em}}c}
  \tabrulename{${\lnot}\nh{\seq}$} &
  \X \seq \nh(A), \Y\\\cmidrule{2-2}
  & \lnot \nh(A), \X \seq \Y
\end{array}
\]
In addition, we have rules for eliminating $\nh$ over negation, based on
equivalence~(\ref{eq-nh-in-not}), and for moving $\nh$ inward, based on
equivalences~(\ref{eq-nh-in-and}) and (\ref{eq-nh-in-or}). The respective
rules are straightforward and omitted from the presentation.

Soundness of \rulename{Ax-nh-1} and \rulename{Ax-nh-2} follows since the
corresponding implications $B \imp (C \lor A \lor \nh(A))$ and $(\lnot A \land
B) \imp (C \lor \nh(A))$ are both valid. Soundness of $\lnot\nh{\seq}$ follows
from the equivalence $(\lnot \nh(A) \land B) \imp C\; \equiv\; B \imp (\nh(A)
\lor C)$.

\paragraph{A Third New Rule for Implication.}

Mints' system handles implication with two rules newly introduced in
\cite{mints}. We introduce a third new rule, which is, like Mints'
${\seq}{\imp}$, for implication in the succedent. We will use it in the
interpolating variation of our sequent system in specific cases, while in
other cases we will stay with Mints' version.
\[
\begin{array}{r@{\hspace{2em}}c@{\hspace{2em}}c}
  \tabrulename{${\seq}{\imp}*$} &
  \X \seq \Y, \nh(A), B & \lnot B, \X \seq \Y, \lnot A\\\cmidrule{2-3}
  & \multicolumn{2}{c}{\X \seq \Y, (A \imp B)}
\end{array}
\]
Soundness of the rule ${\seq}{\imp}*$ follows from the equivalence
\begin{align}
  C \imp (D \lor (A \imp B)) & \xequiv
  (C \imp (D \lor \nh(A) \lor B))
  \land ((\lnot B \land C) \imp (D \lor \lnot A)).
\end{align}

\section{Interpolating Sequent Systems -- Maehara's Method}
\label{sec-maehara-intro}

Our method for Craig-Lyndon interpolation for HT operates by first computing
for two given \htforms $A, B$ a Craig-Lyndon interpolant $C'$ that is a \nhnnf.
This is done with \name{Maehara's method}
\cite{maehara:1960,takeuti:book:1987,smullyan:book:1968,troelstra:schwichtenberg:2000},
a common technique for extracting an interpolant from a sequent proof. We
briefly outline this technique in the context of classical propositional
logic.
As basis we take a sequent system of the G3 family
\cite[Sect.~3.5]{troelstra:schwichtenberg:2000}, where structural rules are
absorbed. Interpolant construction is then specified through axioms and rules
for a generalization of sequents, so-called \defname{split-sequents}, which
have the form
\begin{equation}
  \label{eq-split-sequent}
  \X^\LL, \X^\RR \seqh{H} \Y^\LL, \Y^\RR.
\end{equation}
In a split-sequent each formula occurrence that is a member of the antecedent
or of the succedent (antecedent and succedent are multisets) is labeled with a
\defname{provenance} $\in \{\LL, \RR\}$. We write the multiset of the members
of the antecedent with provenance $\LL$ ($\RR$) as $\X^\LL$ ($\X^\RR$), and,
analogously, the multiset of the members of the succedent with provenance
$\LL$ ($\RR$) as $\Y^\LL$ ($\Y^\RR$). We write $\X$ for $\X^\LL, \X^\RR$ and
$\Y$ for $\Y^\LL, \Y^\RR$.

As usual for G3 systems, axioms and rules are suitable for bottom-up proof
construction that starts with the given goal sequent as root and ends with
axiom instances as leaves. The rules with split-sequents now specify, in
addition, how the provenance labeling of formula occurrences is propagated
upwards from the principal formula in the rule conclusion to the corresponding
active formulas in the premises.

A split-sequent is decorated with a formula $H$, the \name{relative
  interpolant} of the sequent, such that any instance in a proof satisfies the
following conditions.
\begin{align}
  & \textstyle\AX^\LL \entails H \lor \VY^\LL.\label{item-inv-l}\tag{I1}\\
  & \textstyle\AX^\RR \land H \entails \VY^\RR.\label{item-inv-r}\tag{I2}\\
  & \textstyle \vocpn(H) \subseteq \vocpn(\AX^\LL \land \lnot
  \VY^\LL) \cap \vocpn(\lnot \AX^\RR \lor \VY^\RR)\label{item-inv-voc}\tag{I3}
\end{align}
Observe that if $\X = A^\LL$ and $\Y = B^\RR$, then
conditions~(\ref{item-inv-l}--\ref{item-inv-voc}) characterize $H$ as a
Craig-Lyndon interpolant of $A$ and $B$. Thus, a Craig-Lyndon interpolant of
given formulas $A$ and $B$ is obtained as formula $H$ from a proof of the root
sequent
\begin{equation}
  A^\LL \seqh{H} B^\RR.
\end{equation}
Through the decoration with the relative interpolant $H$, the axioms and rules
of the sequent system specify an inductive construction of the Craig-Lyndon
interpolant, which proceeds downwards from leaves, axiom instances, to the
root $A^\LL \seqh{H} B^\RR$. Correctness of the interpolation construction
then means to show that conditions~(\ref{item-inv-l}--\ref{item-inv-voc})
apply to the axioms and, for a rule, if they are assumed for the premises,
then they also hold for the conclusion.

As an example rule consider the interpolating version of ${\seq}{\land}$ for
classical logic (and also HT), where we have two cases, one for each possible
provenance labeling of the principal formula.
\[\begin{array}{l@{\hspace{2em}}c@{\hspace{2em}}c}
\tabrulename{${\seq}{\land}$L} &
\X \seqh{H_1} \Y, A^\LL & \X \seqh{H_2} \Y, B^\LL\\\cmidrule{2-3}
& \multicolumn{2}{c}{\X \seqh{H_1 \lor H_2} \Y, (A \land B)^\LL}
\end{array}
\]
\[\begin{array}{l@{\hspace{2em}}c@{\hspace{2em}}c}
\tabrulename{${\seq}{\land}$R} &
\X \seqh{H_1} \Y, A^\RR & \X \seqh{H_2} \Y, B^\RR\\\cmidrule{2-3}
& \multicolumn{2}{c}{\X \seqh{H_1 \land H_2} \Y, (A \land B)^\RR}
\end{array}
\]
For both rules it is easy to see that syntactic condition~(\ref{item-inv-voc})
transfers from the premises to the conclusion. To show for
\rulename{${\seq}{\land}$L} the transfer of the semantic
condition~(\ref{item-inv-l}) we assume that it holds for the premises, i.e.,
$\AX^\LL \entails H_1 \lor \VY^\LL \lor A$ and $\AX^\LL \entails H_2 \lor
\VY^\LL \lor B$. We can then conclude that it holds for the conclusion:
$\AX^\LL \entails (H_1 \lor H_2) \lor \VY^\LL \lor (A \land B)$.
To show the transfer of the semantic condition~(\ref{item-inv-r}) we assume
$\AX^\RR \land H_1 \entails \VY^\RR$ and $\AX^\RR \land H_2 \entails \VY^\RR$,
and conclude $\AX^\RR \land (H_1 \lor H_2) \entails \VY^\RR$.

\section{An Interpolating Sequent System for HT with NHNNF Interpolants}
\label{sec-interpolating-nhnnf}

We now go into detail of the first stage of our Craig-Lyndon interpolation
method for HT, computing for two given \htforms $A, B$ a Craig-Lyndon
interpolant $C'$ that is a \nhnnf. This is done with an interpolating sequent
system based on a variation of the propositional fragment of Mints' system for
HT from \cite{mints}. In a second stage, from the \nhnnf $C'$ a \htform $C$ is
constructed that is a Craig-Lyndon interpolant of $A, B$. This two-stage
proceeding adapts a technique from \cite{hw:2024}, where a Craig-Lyndon
interpolant of two classical formulas that encode \htforms is postprocessed to
a Craig-Lyndon interpolant that again encodes a \htform. The second stage will
be presented in Sect.~\ref{sec-ht-interpolation}. The result of the first
stage is expressed with the following theorem.
\begin{thm}
  \label{thm-proto-interpolant}
  Let $A,B$ be \htforms such that $A \entails B$. Then there exists
  a \nhnnf $C'$ such that
  \begin{enumerate}
  \item $A \entails C'$ and $C' \entails B$, and
  \item $\vocpn(C') \subseteq \vocpn(A) \cap \vocpn(B)$.
  \end{enumerate}
  Moreover, such a \nhnnf $C'$ can be effectively constructed from a proof of
  $A \entails B$ in a variation of Mints' sequent system for HT.
\end{thm}
We prove this theorem over the rest of this section by specifying an
interpolating sequent system that constructs a suitable \nhnnf $C'$ from given
\htforms $A$ and $B$.

\subsection{An Assumption on the Form of the Second Argument}

\begin{defn}
  \label{def-body-normalize}
  We call a \htform \defname{body-normalized} if no implication occurs in the
  antecedent of another implication.
\end{defn}
We assume that the second argument formula to Craig-Lyndon interpolation, is
\name{body-normalized}. This assumption is w.l.o.g. as any \htform can be
converted to an equivalent one that is body-normalized by rewriting with the
following equivalences.
\begin{align}
  (A\imp B)\imp C & \xequiv (\lnot A\imp C) \land (A \lor \lnot B \lor C) \land (B\imp C).\\
  (A \lor B)\imp C & \xequiv (A\imp C) \land (B\imp C).\\
  (A \land B)\imp C & \xequiv (A\imp C) \lor (B\imp C).\\
  \lnot A \imp B & \xequiv \lnot\lnot A \lor B.
\end{align}  

\subsection{Important Properties of Proofs in the System}

Before presenting the interpolating system, we state properties of all
split-sequents in proofs of a root sequent $A^\LL \seqh{C'} B^\RR$, which
underlie interpolation. This helps in the later presentation of the system
details as it brings their crucial features to attention, although, of course,
technically the stated properties are consequences of these details. Verifying
these properties is not hard but somewhat tedious, typically by inductively
proceeding downwards from leaves to the root, or upwards, from the root to the
leaves. We sketch some of the proofs when presenting axioms and rules in the
context of discussing their crucial features.

\begin{lem}
  \label{lem-h}
  Consider a proof with root sequent
  $A^\LL \seqh{C'} B^\RR$ in the interpolating system of
  Sect.~\ref{sec-sys-axrules}, where $A$ is a \htform and $B$ is a
  body-normalized \htform. The following properties hold for all instances $\X
  \seqh{H} \Y$ of split-sequents in the proof.

  \sublem{lem-h-nhnnf} $H$ is a \nhnnf.

  \sublem{lem-invs} Conditions (\ref{item-inv-l})--(\ref{item-inv-voc})
  (Sect.~\ref{sec-maehara-intro}) are satisfied.\footnote{Although these
  conditions were stated in Sect.~\ref{sec-maehara-intro} for classical logic,
  they apply here if logic operators are understood according to
  Table~\ref{table-truth-tables} and entailment is understood comparative (see
  footnote~\ref{foot-entailment}).}

  \sublem{lem-nh-rr} All occurrences of $\nh$ have provenance $\RR$.

  \sublem{lem-nh-y} All occurrences of $\nh$ are as top-level operators of
  members of\ $\Y$.

  \sublem{lem-nh-arg} Argument formulas of $\nh$ have no occurrences of $\imp$
  and no occurrences of $\nh$.

  \sublem{lem-nopos-rr} All members of\hspace{2pt} $\X^\RR$ are negations,
  i.e., formulas with $\lnot$ as top-level operator.
  
\end{lem}

Theorem~\ref{thm-proto-interpolant} then follows immediately from
Lemmas~\ref{lem-h-nhnnf} and~\ref{lem-invs}.

\subsection{Axioms and Rules of the Interpolating System}
\label{sec-sys-axrules}

\paragraph{Axioms.}

We need the following split-sequent variations of \rulename{Ax-1}
(Sect.~\ref{sec-mints-axioms}).
\[
\begin{array}{l@{\hspace{2em}}l@{\hspace{2em}}l}
  {\small \rulename{Ax-1-LL}} & A^\LL, \X \seqh{\false} \Y, A^\LL & A \text{ a literal}\\
  {\small \rulename{Ax-1-LR}} & A^\LL, \X \seqh{A} \Y, A^\RR & A \text{ a literal}\\
  {\small \rulename{Ax-1-RL}} & (\lnot A)^\RR, \X \seqh{\lnot \lnot A} \Y, (\lnot A)^\LL & A \text{ an atom}\\
  {\small \rulename{Ax-1-RR}} & (\lnot A)^\RR, \X \seqh{\true} \Y, (\lnot
  A)^\RR & A \text{ an atom}
\end{array}
\]
Axioms \rulename{Ax-1-RL} and \rulename{Ax-1-RR} apply only to instances of
axiom~\rulename{Ax-1} where the principal formulas are \emph{negative}
literals. By Lemma~\ref{lem-nopos-rr}, instances where an atom with provenance
$\RR$ is a member of the antecedent are superfluous. For the provenance
labeling of \rulename{Ax-1-RL} this restriction to negative literals is
crucial since for a split-sequent $A^\RR, \X \seqh{H} \Y, A^\LL$ in the case
where $A$ is an atom there actually is no formula $H$ that satisfies the
conditions (\ref{item-inv-l})--(\ref{item-inv-voc}), which characterize the
relative interpolant. `No formula' is quite broad here, even beyond the
considered formula classes: there is no possible logic operator
$\mathit{op}(A)$ that would map atom~$A$ depending on its truth value to
another truth value such that $\mathit{op}(A)$ provides a relative
interpolant~$H$.
For, e.g., \rulename{Ax-1-RL} conditions~(\ref{item-inv-l}) and
(\ref{item-inv-r}) instantiate to $\AX^\LL \entails \lnot\lnot A \lor \VY^\LL
\lor \lnot A$, and $\lnot A \land \AX^\RR \land \lnot \lnot A \entails
\VY^\RR$.

Of axiom \rulename{Ax-2} we need the following split-sequent variations.
\[
\begin{array}{l@{\hspace{2em}}l@{\hspace{2em}}l}
  {\small \rulename{Ax-2-LL}} & A^\LL, (\lnot A)^\LL, \X \seqh{\false} \Y & A \text{ an atom}\\
  {\small \rulename{Ax-2-LR}} & A^\LL, (\lnot A)^\RR, \X \seqh{A} \Y & A \text{ an atom}\\
  {\small \rulename{Ax-2-LR$'$}} & A^\LL, (\lnot A)^\RR, \X \seqh{\lnot \lnot A} \Y & A \text{ an atom}\\
\end{array}
\]
Instances of axiom~\rulename{Ax-2} with provenance $A^\LL$ and $(\lnot A)^\RR$
have both $A$ and also $\lnot \lnot A$ as interpolants. Thus either one of
axioms \rulename{Ax-2-LR} or \rulename{Ax-2-LR$'$} can be chosen
alternatively. 
For, e.g., \rulename{Ax-2-LR}, conditions~(\ref{item-inv-l}) instantiate to $A
\land \AX^\LL \entails A \lor \VY^\LL$ and $\AX^\RR \land \lnot A \land A
\entails \VY^\RR$. For \rulename{Ax-2-LR$'$} they instantiate to $A \land
\AX^\LL \entails \lnot \lnot A \lor \VY^\LL$ and $\AX^\RR \land \lnot A \land
\lnot \lnot A \entails \VY^\RR$.
Axioms \rulename{Ax-nh-1} and \rulename{Ax-nh-2}
(Sect.~\ref{sec-axrules-nh}) are needed in the following split-sequent
variations.
\[
\begin{array}{l@{\hspace{2em}}l@{\hspace{2em}}l}
  {\small \rulename{Ax-nh-1-LR}} & \X \seqh{\nh(A)} \Y, A^\LL, \nh(A)^\RR  & A \text{ an atom}\\
  {\small \rulename{Ax-nh-1-RR}} & \X \seqh{\true} \Y, A^\RR, \nh(A)^\RR  & A \text{ an atom}\\
  {\small \rulename{Ax-nh-2-LR}} & (\lnot A)^\LL, \X \seqh{\lnot A} \Y,
  \nh(A)^\RR  & A \text{ an atom}\\
  {\small \rulename{Ax-nh-2-RR}} & (\lnot A)^\RR, \X \seqh{\true} \Y,
  \nh(A)^\RR  & A \text{ an atom}\\
\end{array}
\]
That other potential variations of the axioms for $\nh$ are superfluous
follows from Lemmas~\ref{lem-nh-rr} and~\ref{lem-nh-y}.
Instances of \rulename{Ax-nh-1-LR} introduce subformulas $\nh(A)$ into the
constructed interpolant. Actually this is the only place where $\nh$ is
introduced into the interpolant. For, e.g., \rulename{Ax-nh-1-LR}
conditions~(\ref{item-inv-l}) and (\ref{item-inv-r}) instantiate to $\AX^\LL
\entails \nh(A) \lor \VY^\LL \lor A$ and $\X^\RR \land \nh(A) \entails \VY^\RR
\lor \nh(A)$.

\paragraph{Implication Rules.}

For provenance $\LL$ of the principal formula, we use Mint's rules
${\imp}{\seq}$ and ${\seq}{\imp}$.
\[
\begin{array}{r@{\hspace{2em}}c@{\hspace{2em}}c@{\hspace{2em}}c}
  \tabrulename{${\imp}{\seq}$L} &
  \lnot A^\LL, \X \seqh{H_1} \Y & \X \seqh{H_2} \Y, A^\LL, (\lnot B)^\LL & B^\LL,\X \seqh{H_3} \Y\\\cmidrule{2-4}
  & \multicolumn{3}{c}{(A \imp B)^\LL, \X \seqh{H_1 \lor H_2 \lor H_3} \Y}
\end{array}
\]
\[
\begin{array}{r@{\hspace{2em}}c@{\hspace{2em}}c}
  \tabrulename{${\seq}{\imp}$L} &
  A^\LL, \X \seqh{H_1} \Y, B^\LL & (\lnot B)^\LL, \X \seqh{H_2} \Y, (\lnot A)^\LL\\\cmidrule{2-3}
  & \multicolumn{2}{c}{\X \seqh{H_1 \lor H_2} \Y, (A \imp B)^\LL}
\end{array}
\]
For provenance $\RR$ of the principal formula, we note that a version of
${\imp}{\seq}$, i.e., for implication with provenance $\RR$ in the antecedent
is by Lemma~\ref{lem-nopos-rr}, superfluous.
For implication with provenance $\RR$ in the succedent we use our new
implication rule ${\seq}{\imp}*$.\footnote{It seems possible to use
${\seq}{\imp}*$ instead of ${\seq}{\imp}$ also for $\LL$-provenance, with
appropriate modifications of other rules and axioms as well as
Lemma~\ref{lem-h}. Our choice of ${\seq}{\imp}$ for $\LL$-provenance
tries to introduce $\nh$ not more than necessary.}
\[
\begin{array}{r@{\hspace{2em}}c@{\hspace{2em}}c}
  \tabrulename{${\seq}{\imp}*$R} &
  \X \seqh{H_1} \Y, \nh(A)^\RR, B^\RR & (\lnot B)^\RR, \X \seqh{H_2} \Y, (\lnot A)^\RR\\\cmidrule{2-3}
  & \multicolumn{2}{c}{\X \seqh{H_1 \land H_2} \Y, (A \imp B)^\RR}
\end{array}
\]

Rule~\rulename{${\seq}{\imp}*$R} is the only rule that, read bottom-up,
introduces $\nh$ into a sequent. The introduced occurrence has provenance
$\RR$, which implies Lemma~\ref{lem-nh-rr}.
The occurrence is placed as a member of the succedent. Lemma~\ref{lem-nh-y}
follows since there is no rule that would move a formula with $\nh$ as
top-level operator from the succedent to an occurrence in the antecedent.
The argument formula of the introduced occurrences of $\nh$ stems from the
antecedent of an implication with provenance $\RR$. Since the given root
succedent $B^\RR$ is body-normalized, it follows that there is no implication
occurring in this argument formula. This implies Lemma~\ref{lem-nh-arg}.

We observe that whenever our rules introduce a formula with provenance $\RR$
into the antecedent of a premise, the formula is a negation. Since the root
sequent has just a single formula with provenance $\LL$ as antecedent, this
implies Lemma~\ref{lem-nopos-rr}.
The statements underlying the transfer of conditions~(\ref{item-inv-l})
and~(\ref{item-inv-r}) from the premises to the conclusion are for the
implication rules as follows.

\newcolumntype{L}{>{$}l<{$}}

\noindent
{\setlength{\arraycolsep}{4pt}
\begin{longtable}{L@{\hspace{12pt}}L}
\text{\rulename{${\imp}{\seq}$L}, \text{(\ref{item-inv-l})}} &
\begin{array}[t]{ll}
\text{If } & \lnot A \land \AX^\LL \entails H_1 \lor  \VY^\LL \text{ and }
 \AX^\LL \entails H_2 \lor \VY^\LL \lor A \lor \lnot B \text{ and}\\
 & B \land \AX^\LL \entails H_3 \lor  \VY^\LL, \text{ then }
   (A \imp B) \land \AX^\LL \entails (H_1 \lor H_2 \lor H_3)  \lor
\VY^\LL.
\end{array}\\[18pt]
\text{\rulename{${\imp}{\seq}$L}, \text{(\ref{item-inv-r})}} &
\begin{array}[t]{ll}
\text{If} & \AX^\RR \land H_1 \entails \VY^\RR \text{ and }
\AX^\RR \land H_2 \entails \VY^\RR \text{ and }
\AX^\RR \land H_3 \entails \VY^\RR, \text{ then}\\
& \hspace{2em} \AX^\RR \land (H_1 \lor H_2 \lor H_3) \entails \VY^\RR.
\end{array}\\[18pt]
\text{\rulename{${\seq}{\imp}$L}, \text{(\ref{item-inv-l})}} &
\begin{array}[t]{ll}
\text{If} & A \land\AX^\LL \entails (H_1 \lor \VY^\LL \lor B) \text{ and }
   \lnot B \land\AX^\LL \entails (H_2 \lor \VY^\LL \lor \lnot A), \text{ then}\\
& \hspace{2em} \AX^\LL \entails (H_1 \lor H_2) \lor \VY^\LL \lor (A\imp B).
\end{array}\\[18pt]
\text{\rulename{${\seq}{\imp}$L}, \text{(\ref{item-inv-r})}} &
\begin{array}[t]{ll}
  \text{If} & \AX^\RR \land H_1 \entails \VY^\RR \text{ and }
     \AX^\RR \land H_2 \entails \VY^\RR, \text{ then}\\
& \hspace{2em} \AX^\RR \land (H_1 \lor H_2) \entails \VY^\RR.
\end{array}\\[18pt]
\text{\rulename{${\seq}{\imp}*$R}, \text{(\ref{item-inv-l})}} &
\begin{array}[t]{ll}
\text{If} & \AX^\LL \entails H_1 \lor \VY^\LL \text{ and }
\AX^\LL \entails H_2 \lor \VY^\LL, \text{ then}\\
& \hspace{2em} \AX^\LL \entails (H_1 \land H_2)  \lor  \VY^\LL.
\end{array}\\[18pt]
\text{\rulename{${\seq}{\imp}*$R}, \text{(\ref{item-inv-r})}} &
\begin{array}[t]{ll}
\text{If} & \AX^\RR \land H_1 \entails \VY^\RR \lor \nh(A) \lor B \text{ and }
\lnot B \land \AX^\RR \land H_2 \entails \VY^\RR \lor \lnot A, \text{ then}\\
& \hspace{2em} \AX^\RR \land (H_1 \land H_2) \entails \VY^\RR  \lor (A\imp B).
\end{array}
\end{longtable}
}

\paragraph{Further Axioms and Rules that are not Presented.}

The following axioms and rules also have to be included in our interpolating
proof system. They are omitted in this presentation as they are
straightforward.

\begin{itemize}
\item Axioms and rules for the truth value constants.

\item Familiar rules for conjunction and disjunction. With the possible
  provenances of the principal formula these are 8 rules. Two of them are
  shown as examples in Sect.~\ref{sec-maehara-intro}.

\item 4 rules for $\lnot \lnot$.

\item 12 rules for moving negation inward.\footnote{As noted in
\cite{otten:schaub:2025}, there is a misprint in \cite{mints}: The correct
premises of the ${\seq}{\lnot\lor}$ rule are $\X \seq \Y, \lnot A$ and $\X
\seq \Y, \lnot B$.}
  
\item 1 rule for eliminating $\nh$ over negation, based on
  equivalence~(\ref{eq-nh-in-not}). By Lemma~\ref{lem-nh-rr}
  and~\ref{lem-nh-y}, this is only needed for the principal formula in the
  succedent and with provenance $\RR$. To effect at bottom-up application a
  reduction of the number of logic operators, required to ensure termination,
  the rule incorporates elimination of the double negation, i.e., it is
  \[\begin{array}{c}
  (\lnot A)^\RR, \X \seqh{H} \Y\\\midrule \X \seqh{H} \Y, \nh(\lnot A)^\RR
  \end{array}
  \]

\item 2 rules for moving $\nh$ inward, based on
  equivalences~(\ref{eq-nh-in-and}) and~(\ref{eq-nh-in-or}). By
  Lemma~\ref{lem-nh-rr} and~\ref{lem-nh-y}, these are only needed for the
  principal formula in the succedent and with provenance $\RR$.

\item By Lemma~\ref{lem-nh-y} the rule $\lnot\nh{\seq}$ is superfluous.
  
\end{itemize}

\subsection{Completeness}
\label{sec-completeness}

Completeness of our interpolating sequent system can be shown similarly as
outlined by Mints \cite{mints}. Termination is ensured since, applied
bottom-up, each of the rules either reduces the number of occurrences of
binary logic operators or keeps that number but reduces the number of
occurrences of unary operators.
We note that whenever $\X \seqh{H} \Y$ is a split-sequent that meets the
conditions of Lemmas~\ref{lem-nh-rr}--\ref{lem-nopos-rr} and $\X$ or $\Y$ has
a member that is not of the form $A$, $\lnot A$, or $\nh(A)$, with $A$ an
atom, then some rule is applicable bottom-up to the split-sequent. For a leaf
sequent $\X \seqh{H} \Y$ to which no rule is applicable and which is no
instance of an axiom we can construct a refuting model by assigning truth
values to all atoms $A$ as follows.
\begin{center}
\begin{tabular}{l}
  If $A \in \X$, then $A := \T$,\\
  else if $\nh(A) \in \Y$, then $A := \T$,\\
  else if $\lnot A \in \Y$, then $A := \NF$,\\
  else $A := \F$.
\end{tabular}
\end{center}
We obtain a truth value assignment under which $\AX$ has the value $\T$, while
$\VY$ has the value $\NF$ or the value $\F$. Hence, the value of $\AX \imp
\VY$ must be $\F$ or $\NF$ but cannot be $\T$.

Sine all our rules are invertible and sound, the formula $A \imp B$ for the
two interpolation arguments $A$ and $B$ is equivalent to the conjunction of
$\bigwedge \X \imp \bigvee \Y$ for all leaves of the proof tree that are not
closed by an axiom. Thus, if one of these conjuncts has the value $\F$ or
$\NF$, it follows that the value of $A \imp B$ can only be $\F$ or $\NF$ but
not $\T$.

\section{Craig-Lyndon Interpolation for HT}
\label{sec-ht-interpolation}

We now turn our attention to the second stage of our Craig-Lyndon
interpolation method for HT, the conversion of the interpolating \nhnnf
formula $C'$ obtained according to Theorem~\ref{thm-proto-interpolant} to a
Craig-Lyndon interpolant $C$ that is a \htform. This conversion is based on
the following theorem.
\begin{thm}
  \label{thm-strengthening}
  Let $A$ be a \htform and let $C'$ be a \nhnnf such that $A \entails C'$.
  Then there exists a \htform~$C$ such that $A \entails C$, $C \entails C'$
  and $\vocpn(C) \subseteq \vocpn(C')$. Moreover, such a \htform $C$ can be
  effectively constructed from $C'$.
\end{thm}
\begin{proof}
  We show the construction of a suitable \htform $C$ from the \nhnnf $C'$.
  First $C'$ is converted to an equivalent CNF, which can be done with the
  familiar conversions for classical propositional logic. Let $D = \nh(E_1)
  \lor \ldots \lor \nh(E_m) \lor F$, with $m \geq 0$, where $F$ has no
  occurrences of $\nh$, be an arbitrary clause of this CNF. It holds that $A
  \entails D$. Hence $\lnot \lnot A \entails \lnot \lnot D$. Hence there is a
  proof $P$ of $\lnot \lnot A \entails \lnot \lnot D$ in Mints' system
  extended with axiom~\rulename{Ax-nh-2} and rule~$\lnot\nh{\seq}$ defined in
  Sect.~\ref{sec-axrules-nh} (completeness of this system for the considered
  implication formulas follows with the model construction of
  Sect.~\ref{sec-completeness}). Note that axiom~\rulename{Ax-nh-1} is not
  required: since the root sequent is $\lnot \lnot A \seq \lnot \lnot D$, if
  no rule is applicable to a sequent $\X \seq \Y$, then each member of $\X,\Y$
  is either of the form $\lnot p$ or of the form $\nh(p)$, with $p$ an atom.
  Let $D'$ be $D$ with all formulas of the form $\nh(E_i)$ replaced by $\lnot
  E_i$. We modify the proof $P$ to a proof $P'$ of $\lnot \lnot A \entails
  \lnot \lnot D'$ in Mints' original system as follows: (1) Throughout the
  proof replace all formulas of the form $\nh(E_i)$ with $\lnot E_i$.
  (2) Readjust rule labels: Leaves that were an instance of \rulename{Ax-nh-2}
  are now instances of \rulename{Ax-1}; applications of rule
  \rulename{${\lnot}\nh{\seq}$} become applications of Mints' rule
  \rulename{${\lnot}{\lnot}{\seq}$}. Proof $P'$ thus justifies $\lnot \lnot A
  \entails \lnot \lnot D'$. Since $A \entails \lnot \lnot A$, it follows that
  $A \entails \lnot \lnot D'$. Thus, $A \entails D \land \lnot \lnot D'$. The
  entailed conjunction $D \land \lnot \lnot D'$ is
  \begin{align}(\nh(E_1) \lor \ldots \lor \nh(E_m) \lor F)
    \land \lnot \lnot (\lnot E_1 \lor \ldots \lor \lnot E_m \lor F),
  \end{align}
  which is equivalent to
  \begin{align}(\nh(E_1) \lor \ldots \lor \nh(E_m) \lor F)
    \land (\lnot E_1 \lor \ldots \lor \lnot E_m \lor \lnot \lnot F),
  \end{align}
  which is equivalent to the \htform
  \begin{align}E_1 \land \ldots \land E_m \imp F.
  \end{align}
 Hence, the conjunction of the implications $E_{1} \land \ldots \land E_{m_j}
 \imp F_j$ for all clauses $D_j = \nh(E_{1}) \lor \ldots \lor \nh(E_{m_j}) \lor
   F_j$ of the CNF of $C'$ provides the \htform $C$ as claimed in the
   theorem statement.
   \qed
\end{proof}

Although the proof of Theorem~\ref{thm-strengthening} refers to both
formulas $A$ and $C'$, the \emph{construction} of $C$ is actually just from
$C'$, independently from $A$. The sequent proofs $P$ and $P'$ referenced in
the proof of the theorem are for the purpose of justifying properties
of~$C$ but not used in the actual construction of~$C$.

Theorem~\ref{thm-strengthening} is in essence an adaptation of the final step
in the proof of Theorem~11 in \cite{hw:2024}. In this step, an
LP-interpolant~$H$ that encodes a logic program is obtained from a
Craig-Lyndon interpolant $H'$ for classical formulas that encode logic
programs as $H := H' \land \f{rename}_{0 \mapsto 1}(H')$.

\pagebreak
We are now ready to state the overall result, Craig-Lyndon interpolation for HT on
the basis of Mints' sequent system.
\begin{thm}
  \label{th-interpolant}
  Let $A,B$ be \htforms such that $A \entails B$. Then there exists a \htform
  $C$ such that
  \begin{enumerate}
  \item $A \entails C$ and $C \entails B$, and
  \item $\vocpn(C) \subseteq \vocpn(A) \cap \vocpn(B)$.
  \end{enumerate}
  Moreover, such a \htform $C$ can be effectively constructed from a proof of
  $A \entails B$ in a variation of Mints' sequent system for HT.
\end{thm}

\begin{proof} Follows from Theorems~\ref{thm-proto-interpolant} and~\ref{thm-strengthening}.
\qed
\end{proof}

\section{Conclusion}
\label{sec-conclusion}

We have introduced a method to construct Craig-Lyndon interpolants for HT,
which proceeds in two stages, extraction of a preliminary interpolant in an
extended logic with Maehara's method, followed by a conversion to an
interpolant in HT.
The method is an adaptation of an encoding-based interpolation method for HT
implemented by automated first-order provers \cite{hw:2024}.
Our sequent calculus reformulation aims at linking to research threads in
proof theory and supporting potential applications where the focus are the
proofs themselves, for example as explanations, in contrast to the
interpolating formulas.

A recent observation \cite{saurin:proofrelevant} relates to the focus on
proofs: Maehara's method can be refined to \name{proof relevant
  interpolation}, a way of cut introduction, or proof splitting. As our
preliminary interpolant is obtained with Maehara's method, this refinement
should transfer to our first stage. Incorporation and apparent necessity of
the second stage, the conversion to HT, has to be investigated.

For HT, the distinction between \name{positive} and \name{negative} polarity
underlying Craig-Lyndon interpolation suggests refinement. If a HT formula is
normalized to a ``logic program'' (a conjunction of implications $A \imp B$,
with $A$ a conjunction of literals and $B$ a disjunction of literals), then
\emph{four} roles may be distinguished for an atom occurrence: positively in
$A$, negated in $A$, positively in $B$, and negated in $B$. It seems that
roles of atom occurrences in interpolants can be ensured to be in certain
relationships with the roles in the interpolated formulas, similarly as
discussed for Beth definability in \cite[Sect.~3.4]{hw:2024}.

Two steps in our interpolation method for HT involve formula transformations
that seem not to be polynomially restricted: Lemma~\ref{lem-h} requires the
second interpolation input to be \name{body-normalized}, and the
postprocessing according to Theorem~\ref{thm-strengthening} involves
transformation to CNF. The actual impact of this is not clear. Logic programs
are typically already in body-normalized form. Possibly there are more
sensible transformations such as a relaxed body-normalized condition and an
adaptation of the Tseitin encoding. For methods avoiding non-polynomial
formula transformations, the question arises whether effort has just been
passed to proving, with proof steps now performing the work of the
transformation.

\begin{acknowledgments}
The author thanks David Pearce for the conversations that led to this work.
Funded by the Deutsche Forschungsgemeinschaft (DFG, German Research
Foundation) -- Project-ID~457292495.
\end{acknowledgments}

\section*{Declaration on Generative AI}
The author(s) have not employed any Generative AI tools.

\bibliography{bibcraig}

\closeout\plabelsfile
\closeout\plabelslogfile

\end{document}